\newif\ifjcgs
\jcgstrue
\jcgsfalse

\ifjcgs
  \documentclass[12pt]{article}
\else
  \documentclass{article}
\fi

\usepackage{arxiv}
\renewcommand{\headeright}{}
\renewcommand{\undertitle}{}
\usepackage[T1]{fontenc}
\usepackage[utf8]{inputenc}
\usepackage{amsmath}
\usepackage{amssymb}
\usepackage{amsthm}
\usepackage{graphicx}
\usepackage{algorithm}
\usepackage{booktabs}
\usepackage{multirow}
\usepackage{dcolumn}
\usepackage[round]{natbib}
\usepackage{microtype}
\usepackage{setspace}
\usepackage{xcolor}
\usepackage{hyperref}
\hypersetup{
  pdftitle={Direct Simulation of Poisson and Newton Kernel Distributions},
  pdfkeywords={Directional statistics; Exact sampling; Power-Saw family; Random chord; Spherical distribution},
  colorlinks=true,
  linkcolor=blue,
  citecolor=blue,
  urlcolor=blue
}

\hypersetup{pdfauthor={Lukas Sablica; Kurt Hornik}}

\newtheorem{lemma}{Lemma}
\newtheorem{theorem}{Theorem}
\newtheorem{corollary}{Corollary}

\newcounter{algln}
\newcommand{\algkw}[1]{\textup{\textbf{#1}}}
\newenvironment{algsteps}
  {\setcounter{algln}{0}%
   \begin{list}{\textup{\arabic{algln}.}}%
     {\usecounter{algln}%
      \setlength{\leftmargin}{2.0em}\setlength{\labelwidth}{1.4em}%
      \setlength{\labelsep}{0.6em}\setlength{\itemsep}{2pt}%
      \setlength{\topsep}{3pt}\setlength{\parsep}{0pt}}}
  {\end{list}}

\newcommand{\Sd}{\mathbb{S}^{d-1}}
\newcommand{\R}{\mathbb{R}}
\newcommand{\norm}[1]{\lVert #1\rVert}
\newcommand{\tr}{^{\top}}

\title{Direct Simulation of Poisson and Newton Kernel Distributions}

\author{
  \normalfont
  \begin{minipage}[t]{0.47\textwidth}
    \centering
    \textbf{Lukas Sablica}\\
    {\fontsize{10}{12}\selectfont
      Institute for Statistics and Mathematics\\
      WU Vienna University of Economics and Business\\
      Austria}\\
    ORCID: 0000-0001-9166-4563\\
    \href{mailto:Lukas.Sablica@wu.ac.at}{\nolinkurl{Lukas.Sablica@wu.ac.at}}%
    \ifjcgs\\Corresponding author\fi
  \end{minipage}%
  \hspace{0.02\textwidth}%
  \begin{minipage}[t]{0.47\textwidth}
    \centering
    \textbf{Kurt Hornik}\\
    {\fontsize{10}{12}\selectfont
      Institute for Statistics and Mathematics\\
      WU Vienna University of Economics and Business\\
      Austria}\\
    ORCID: 0000-0003-4198-9911\\
    \href{mailto:Kurt.Hornik@wu.ac.at}{\nolinkurl{Kurt.Hornik@wu.ac.at}}
  \end{minipage}
}

\begin{document}
\raggedbottom
\maketitle

\begin{abstract}
The Poisson kernel-based (PKB) distribution is a flexible and numerically convenient
family on the unit sphere used for model-based clustering, uniformity testing,
and directional data analysis. Available random variate generators rely on
rejection from a von Mises--Fisher, angular central Gaussian, or projected Saw
envelope and therefore require a rejection loop, tuning or root-finding, or a
univariate setup phase. We show that none is needed. A uniformly random line
through the interior point $\rho\mu$ meets the sphere in two points; selecting
one with probability proportional to the opposite segment length returns
exactly the target density. The sampler is exact, rejection- and tuning-free,
uses one uniform direction and one Bernoulli draw, and costs $O(d)$ operations.
We place the distribution in the power-Saw family of densities proportional to
$\norm{x-\rho\mu}^{-\xi}$ and show that, for $d\geq3$, the same construction is
exact for one further exponent. This introduces the Newton kernel-based (NKB)
distribution, whose normalizing constant is identically one. Across a wide grid
of dimensions, concentrations, and sample sizes, the direct sampler is faster
than the better of two previously recommended methods in every case, with the
largest gains at moderate concentration.
\end{abstract}


\section{Introduction}

Directional data, observations on the unit sphere $\Sd=\{x\in\R^d:\norm{x}=1\}$,
arise throughout the sciences, and a growing methodology models them with
rotationally symmetric distributions \citep{MardiaJupp2000}.
\citet{GolzyMarkatou2020} introduced the Poisson kernel-based (PKB) distribution for
model-based clustering on the sphere. Unlike most spherical distributions, its
density is elementary and free of special functions, which makes it especially
attractive for mixture modeling and has led to its rapid adoption in software
and methodology for spherical clustering, uniformity testing and directional
data analysis \citep{DingMarkatouSaraceno2025, SaracenoMukhopadhyayMarkatou2025,
SablicaHornikGruen2025, TsagrisPapastamoulisKato2025}. The same numerical
simplicity that recommends the distribution for modeling also makes a fast and
reliable generator valuable, much as von Mises-Fisher mixtures became practical
once their fitting and simulation were automated \citep{HornikGrun2014}.

Such applications need a generator. Random draws underpin parametric bootstrap
inference \citep{EfronTibshirani1993}, including the likelihood ratio test for
the number of mixture components \citep{McLachlan1987}, Bayesian predictive
simulation, and Monte Carlo evaluation of estimators. The quality of the
generator, in particular in high dimensions, is therefore of direct practical
interest.

The PKB distribution is, however, not trivial to sample.
\citet{GolzyMarkatou2020} proposed rejection from a von Mises-Fisher envelope,
and \citet{SablicaHornikLeydold2023} developed and optimized rejection samplers
for the PKB distribution. They showed that the von Mises-Fisher envelope is not
uniformly efficient as the concentration approaches its upper limit, optimized
the envelope, and gave two superior alternatives based on angular central Gaussian envelopes
\citep{Tyler1987} and on adaptive rejection sampling from a projected Saw
distribution \citep{Saw1978, Ulrich1984}. All of these are rejection methods.
Each pays for its generality with a rejection loop whose expected cost exceeds
one, a tuning or root-finding step, or a univariate set-up phase. The general
theory of such methods is given in \citet{Devroye1986}. 

There is little reason to incur any of these costs. The present work takes a fundamentally different approach and shows that the Poisson kernel-based distribution can be sampled directly. Let $a$ denote the interior point of the unit ball determining the distribution. A single uniformly random line
through $a$ meets the sphere in two points, and if one of them is chosen with
probability proportional to the length of the opposite segment, the chosen
point has exactly the target density. Every draw is accepted, no envelope is
constructed, and the procedure reduces to one spherical-uniform draw, one
Bernoulli draw and $O(d)$ arithmetic. Thus the contribution is a conceptual replacement
of the earlier algorithms, not merely their next speed improvement. This places the PKB distribution on the same footing as the spherical Cauchy distribution, which alone among the standard
spherical models already admits an equally simple sampler, through a M\"obius
transformation of a uniform point \citep{KatoMcCullagh2020}.

The chord construction also has a classical potential-theoretic background.
Malmheden's theorem and modern treatments relate chord averages in balls to the
Dirichlet problem \citep{Malmheden1934,AgranovskyKhavinsonShapiro2010}. Here we
use the same elementary line geometry for exact random generation on the sphere.

Two further points organize the paper. First, the construction is not an
isolated trick. It belongs to the power-Saw family of spherical densities
proportional to $\norm{x-\rho\mu}^{-\xi}$, of which the PKB
distribution is the member $\xi=d$. Within this family the uniform-line
construction is exact for exactly one other exponent, $\xi=d-2$, when $d\geq3$, and we use that
exponent to introduce a companion family, the Newton kernel-based (NKB) distribution,
whose normalizing constant equals one for every concentration and dimension $d\geq3$.
Second, the direct sampler is not merely simpler in principle. In a simulation
study spanning dimensions from $3$ to $1000$, concentrations from $0.01$ to
$0.99$ and sample sizes from $10$ to $10^4$, it is faster than the better of the
two state-of-the-art methods in every cell, with the largest
gains precisely where those methods are individually weakest.

\section{The Poisson kernel-based distribution and the power-Saw family}

Let $\upsilon_d$ denote the uniform probability measure on $\Sd$. For
$0\le\rho<1$ and $\mu\in\Sd$, the PKB distribution has density
\begin{equation}\label{eq:pkbd}
  f(x\mid\rho,\mu)=\frac{1-\rho^2}{\norm{x-\rho\mu}^{d}},
  \qquad x\in\Sd,
\end{equation}
with respect to $\upsilon_d$. At $\rho=0$ it is uniform, and as $\rho\to1-$ it
concentrates at $\mu$. The single interior point $a=\rho\mu$ controls both the
mean direction $\mu$ and the concentration $\rho=\norm{a}$.

\citet{Saw1978} studied the rotationally symmetric densities on $\Sd$
proportional to $g(\lambda\mu\tr x)$, where $g\ge0$ is a shape function and
$\lambda\ge0$ a concentration parameter. The von Mises-Fisher and Watson
families are familiar instances. Writing
\begin{equation}\label{eq:reparam}
  \norm{x-\rho\mu}^2=1+\rho^2-2\rho\,\mu\tr x
  =(1+\rho^2)\bigl(1-\lambda\,\mu\tr x\bigr),
  \qquad
  \lambda=\frac{2\rho}{1+\rho^2},
\end{equation}
shows that \eqref{eq:pkbd} is a Saw density with shape $g(t)=(1-t)^{-d/2}$, mean
direction $\mu$ and concentration $\lambda=\lambda(\rho)$. The map
$\rho\mapsto\lambda$ increases from $0$ to $1$ on $[0,1)$, so the concentration
may be read from either parameter.

More generally, for an exponent $\xi>0$ consider the densities proportional to
$\norm{x-\rho\mu}^{-\xi}$, with normalizing constant
\begin{equation}\label{eq:Cxi}
  C_\xi(\rho)=\int_{\Sd}\norm{x-\rho\mu}^{-\xi}\,d\upsilon_d(x),
\end{equation}
which is independent of $\mu$ by rotational invariance. By \eqref{eq:reparam}
these are Saw densities with shape $g(t)=(1-t)^{-\xi/2}$, and we call the
resulting collection the power-Saw family. The PKB
distribution is the member $\xi=d$, with $C_d(\rho)=(1-\rho^2)^{-1}$. The family
is of independent interest. It is the exit distribution of Brownian motion from
the ball for $\xi=d$ \citep{Durrett1984, KatoJones2013}, and it contains the
spherical Cauchy distribution at $\xi=2(d-1)$ \citep{KatoMcCullagh2020}. The
spherical Cauchy distribution can be sampled by a M\"obius transformation of a
uniform point, but no comparably direct generator has been available for the
Poisson exponent. The next section supplies one, and Section~\ref{sec:newton} identifies the
single further exponent for which the same idea is exact.

\section{A direct sampler}

Fix $a=\rho\mu$ with $0<\rho<1$. The case $\rho=0$ is uniform. For a direction
$v\in\Sd$, the line $a+\R v$ meets the sphere in the two points
\begin{equation}\label{eq:endpoints}
  x_+(v)=a+r_+(v)\,v,
  \qquad
  x_-(v)=a-r_-(v)\,v,
\end{equation}
where, with $c(v)=a\tr v$ and $\Delta(v)=\{c(v)^2+1-\rho^2\}^{1/2}$,
\begin{equation}\label{eq:radii}
  r_+(v)=\Delta(v)-c(v),
  \qquad
  r_-(v)=\Delta(v)+c(v).
\end{equation}
Both are positive, and
\begin{equation}\label{eq:rsident}
  r_+(v)\,r_-(v)=1-\rho^2,
  \qquad
  r_+(v)+r_-(v)=2\Delta(v).
\end{equation}
The construction rests on two facts. The first is the density of the point at
which a uniform ray from $a$ hits the sphere.

\begin{lemma}\label{lem:solid}
Let \(V\sim\upsilon_d\) and let \(X=a+r_+(V)V\) be the point at which the ray
from \(a\) in direction \(V\) meets \(\Sd\). Then \(X\) has density
\begin{equation}\label{eq:solid}
  q(x\mid\rho,\mu)=\frac{1-a\tr x}{\norm{x-a}^{d}},
  \qquad x\in\Sd,
\end{equation}
with respect to $\upsilon_d$.
\end{lemma}

\begin{proof}
By rotational symmetry about $\mu$ it suffices to track the marginal along
$\mu$. Put $T=\mu\tr V$ and $U=\mu\tr X$. For a uniform point on the sphere $T$
has density $g_d(t)=c_d(1-t^2)^{(d-3)/2}$ on $(-1,1)$ \citep{MardiaJupp2000}, with
$c_d=\Gamma(d/2)/\{\sqrt{\pi}\,\Gamma((d-1)/2)\}$. Writing $r=\norm{X-a}$ and
using $V=(X-a)/r$,
\[
  T=\frac{U-\rho}{r}=\frac{U-\rho}{(1+\rho^2-2\rho U)^{1/2}}.
\]
Differentiation gives $dT/dU=(1-\rho U)/r^3$ and $1-T^2=(1-U^2)/r^2$. By the
change of variables formula, the density of $U$ induced by a uniform $V$ is
\[
  g_d(T(U))\,\Bigl|\frac{dT}{dU}\Bigr|
  =c_d(1-U^2)^{(d-3)/2}\,\frac{1-\rho U}{r^{d}}.
\]
Dividing by $g_d(U)$, the uniform marginal of $\mu\tr X$, gives the density of
$X$ with respect to $\upsilon_d$, which at $x$ with $U=\mu\tr x$ and
$r=\norm{x-a}$ equals \((1-\rho U)/r^{d}=(1-a\tr x)/\norm{x-a}^{d}.\)
\end{proof}

The factor \((1-a^\top x)/\norm{x-a}^{d}\) is the usual differential
solid-angle factor for radial projection from the interior point \(a\) onto
the unit sphere. In the notation of \citet{SlumpHoenders1985}, their formula is
\((\boldsymbol n^\top\boldsymbol r)/\|\boldsymbol r\|^d\,d\sigma\), where
\(\boldsymbol r\) is the vector from the projection center to the surface
element. Taking \(\boldsymbol r=x-a\) and \(\boldsymbol n=x\) gives the factor
above.

The second fact is a chord identity. For $x$ on the sphere with $r=\norm{x-a}$
and $s$ the distance from $a$ to the opposite intersection point, $rs=1-\rho^2$
by \eqref{eq:rsident}, so
\begin{equation}\label{eq:chord}
  r(r+s)=r^2+rs=(1+\rho^2-2a\tr x)+(1-\rho^2)=2(1-a\tr x).
\end{equation}
Selecting $x$ against the opposite segment therefore contributes the factor
\begin{equation}\label{eq:select}
  \frac{s}{r+s}=\frac{rs}{r(r+s)}=\frac{1-\rho^2}{2(1-a\tr x)}.
\end{equation}

\begin{theorem}\label{thm:pkbd}
Let $V\sim\upsilon_d$, form $x_+(V)$ and $x_-(V)$ by
\eqref{eq:endpoints}-\eqref{eq:radii}, and given $V$ return $X=x_+(V)$ with
probability $r_-(V)/\{r_+(V)+r_-(V)\}$ and $X=x_-(V)$ otherwise. Then $X$ has the
Poisson kernel-based density \eqref{eq:pkbd}.
\end{theorem}

\begin{proof}
A point $x$ is reached as $x_+(v)$ when $V=v$ and as $x_-(-v)$ when $V=-v$, and
both directions are equally likely under $\upsilon_d$, contributing a factor
$2$. With $r=\norm{x-a}$ and $s$ as above, \eqref{eq:chord} gives
$1-a\tr x=r(r+s)/2$, so the solid-angle density \eqref{eq:solid} equals
$q(x\mid\rho,\mu)=(r+s)/(2r^{d-1})$. For $x=x_+(v)$ the own and opposite segments
are $r=r_+(v)$ and $s=r_-(v)$, so the selection probability
$r_-(v)/\{r_+(v)+r_-(v)\}$ equals \eqref{eq:select}. Multiplying the three
contributions,
\[
  2\cdot\frac{r+s}{2r^{d-1}}\cdot\frac{s}{r+s}
  =\frac{s}{r^{d-1}}
  =\frac{1-\rho^2}{r^{d}}
  =\frac{1-\rho^2}{\norm{x-a}^{d}},
\]
using $s=(1-\rho^2)/r$. This is \eqref{eq:pkbd}.
\end{proof}

The sampler is summarized in Algorithm~\ref{alg:pkbd}. The uniform direction is
obtained as a normalized Gaussian vector, the standard device for the spherical
uniform distribution \citep{HormannLeydoldDerflinger2004}. The cost is one such
vector, one uniform variate and $O(d)$ arithmetic, and there is no rejection
loop and no precomputation. For numerical stability near \(\rho=1\), the shorter
segment should not be computed by subtracting nearly equal numbers. Step 3 of
Algorithm~\ref{alg:pkbd} instead uses \(r_+r_-=1-\rho^2\) explicitly.

\begin{algorithm}[H]
\caption{Direct sampler for the PKB distribution with
parameters $\rho$ and $\mu$ in dimension $d$}
\label{alg:pkbd}
\begin{algsteps}
\item Draw $Z\sim N_d(0,I_d)$ and set $V=Z/\norm{Z}$.
\item Set $M=(1-\rho)(1+\rho)$, $C=\rho\,\mu\tr V$, and
$\Delta=(C^2+M)^{1/2}$.
\item \algkw{if} $C\geq0$ \algkw{then} set $r_-=\Delta+C$ and $r_+=M/r_-$;
\algkw{otherwise} set $r_+=\Delta-C$ and $r_-=M/r_+$.
\item Draw $U\sim\mathrm{Un}(0,1)$, independently of $V$.
\item \algkw{if} $U\le r_-/(r_++r_-)$ \algkw{then return} $X=\rho\mu+r_+V$,
\algkw{otherwise return} $X=\rho\mu-r_-V$.
\end{algsteps}
\end{algorithm}

\section{The Newton kernel-based distribution}
\label{sec:newton}

The argument of Theorem~\ref{thm:pkbd} used the exponent $\xi=d$ only at the
last step. Carrying a general exponent through the same computation shows when
the construction stays exact. Put $\alpha=d-1-\xi$. If the endpoint is selected
with probability proportional to a power $r^\alpha$ of its own distance, the
density induced at $x$ is proportional to $r^{\alpha-d+1}$ divided by the
chord-direction weight $\{r_+(v)^\alpha+r_-(v)^\alpha\}/\{r_+(v)+r_-(v)\}$
evaluated along the line. Uniform directions are exact precisely when this
weight is constant.

\begin{corollary}\label{cor:two}
For $d\geq3$, within the power-Saw family and for $0< \rho <1$, the uniform-line construction with endpoint selection proportional to a fixed power of the segment length is exact and free
of rejection for exactly two exponents, $\xi=d$ and $\xi=d-2$. The case $\xi=d$
is the Poisson kernel-based distribution, with selection proportional to the
opposite segment. The case $\xi=d-2$ has selection proportional to the own
segment, and its normalizing constant is identically one,
\begin{equation}\label{eq:newtonconst}
  \int_{\Sd}\norm{x-\rho\mu}^{-(d-2)}\,d\upsilon_d(x)=1,
  \qquad 0\le\rho<1.
\end{equation}
\end{corollary}

For $d=2$, the second exponent would be $\xi=d-2=0$, which lies outside the
power-Saw family as defined above for $\xi>0$ and gives the uniform density.
Thus the second case in two dimensions is a uniform degeneracy rather than a
second nondegenerate power-Saw member.

\begin{proof}
Let \(m=1-\rho^2\) and let \(s=m/r\) be the opposite segment length. The
chord-direction weight is \(H_\alpha(r)=(r^\alpha+m^\alpha r^{-\alpha})/
(r+m r^{-1})\). Since \(r\) ranges over a nondegenerate interval, uniform lines
are exact only if \(H_\alpha(r)\) is constant there. If \(H_\alpha(r)\equiv c\),
then \(r^\alpha+m^\alpha r^{-\alpha}=c r+c m r^{-1}\). The functions \(r^\beta\) with distinct exponents \(\beta\) are linearly
independent on any interval. Therefore the two exponents on the left,
\(\alpha\) and \(-\alpha\), must be exactly the two exponents on the right,
\(1\) and \(-1\). Thus \(\alpha=\pm1\). For \(\alpha=-1\),
that is \(\xi=d\), \(H_\alpha(r)=1/m\), and selection proportional to \(r^{-1}\)
is selection proportional to the opposite segment, recovering
Theorem~\ref{thm:pkbd}. For \(\alpha=1\), that is \(\xi=d-2\),
\(H_\alpha(r)=1\), and selection is proportional to the own segment. Therefore
\(C_{d-2}(\rho)=1\), giving \eqref{eq:newtonconst}. Indeed, the construction
itself defines a probability law, and in the case $\alpha=1$ the endpoint
density produced by the calculation is exactly $r^{-(d-2)}$, without an
additional $\rho$-dependent factor. Its integral must therefore equal one,
which establishes the normalization directly from the chord construction.
\end{proof}

We call the $\xi=d-2$ member the Newton kernel-based (NKB) distribution, with density
$\norm{x-\rho\mu}^{-(d-2)}$ with respect to $\upsilon_d$. Two features make it a
natural companion to the PKB distribution. The kernel
$\norm{x-a}^{-(d-2)}$ is, up to a constant, the Newtonian potential in $\R^d$,
and \eqref{eq:newtonconst} is the classical statement that a uniform spherical
shell produces a constant potential in its interior, equal to the value on the
surface. The density therefore carries no normalizing constant at all, for every
$\rho$ and every $d\ge3$. It is generated by the same one-line procedure as the
Poisson case, the only change being that the endpoint is chosen in proportion to
its own distance from $a$ rather than the opposite one, so the cost is again one
uniform direction, one Bernoulli draw and $O(d)$ arithmetic. 

The limiting behavior is nevertheless quite different from that of the
PKB distribution. As \(\rho\to1^-\), the NKB density tends to
the cusp \(\norm{x-\mu}^{-(d-2)}\), which is singular at \(\mu\) but still
integrable on \(\Sd\). Thus increasing \(\rho\) sharpens the density without
producing a Dirac limit. The same pointwise singularity also prevents ordinary
unrestricted maximum likelihood from being well behaved. For \(d>2\), moving
\(a=\rho\mu\) towards any observation \(x_i\), say \(a=r x_i\) with
\(r\to1^-\), sends the likelihood contribution of that observation to infinity,
and there is no normalizing constant depending on \(\rho\) to offset it. For
\(d=2\) the exponent is zero and the family is uniform. We therefore regard the
NKB law mainly as a structurally useful companion to the PKB
distribution, and potentially as a regularized or bounded-concentration kernel,
rather than as a drop-in replacement for it. The spherical Cauchy exponent
\(\xi=2(d-1)\) is not one of the two exact uniform-line exponents, so this
construction also does not supersede the M\"obius sampler of
\citet{KatoMcCullagh2020}.

\section{Numerical study}

We repeat the simulation design of \citet{SablicaHornikLeydold2023}. For each
combination of dimension $d\in\{3,\allowbreak5,\allowbreak10,\allowbreak20,
\allowbreak50,\allowbreak100,\allowbreak200,\allowbreak1000\}$, concentration
$\rho\in\{0.01,\allowbreak0.1,\allowbreak0.25,\allowbreak0.4,\allowbreak0.6,
\allowbreak0.75,\allowbreak0.9,\allowbreak0.99\}$ and sample size
$n\in\{10,\allowbreak100,\allowbreak1000,\allowbreak10^4\}$ we generated $n$ variates and recorded the average
generation time over $300$ replications. Our new direct sampling method,
Algorithm~\ref{alg:pkbd}, is compared against the better, in each cell, of the two
methods recommended by \citet{SablicaHornikLeydold2023}: the angular central
Gaussian rejection sampler and the projected Saw sampler that uses the Tinflex
adaptive rejection algorithm \citep{GilksWild1992, BottsHormannLeydold2013,
LeydoldTinflex2021}. All three samplers were implemented in \texttt{C++} through Rcpp
\citep{EddelbuettelFrancois2011} and called from R
\citep{RCoreTeam}. The angular central Gaussian and Tinflex/Saw implementations
were taken from the \texttt{circlus} package \citep{SablicaHornikGruen2025}, and
the proposed direct sampler is available in the
\href{https://github.com/lsablica/rPKB}{\texttt{rPKB}} package \citep{rPKB}.
All benchmarks were run on the WU cluster under Ubuntu 22.04.5 using R 4.5.1
and GCC 12.2.0. Figure~\ref{fig:speedup} reports the speed-up factor, the
ratio of the better competing time to the time of Algorithm~\ref{alg:pkbd}.

\begin{figure}[!htbp]
\centering
\includegraphics[width=0.94\textwidth]{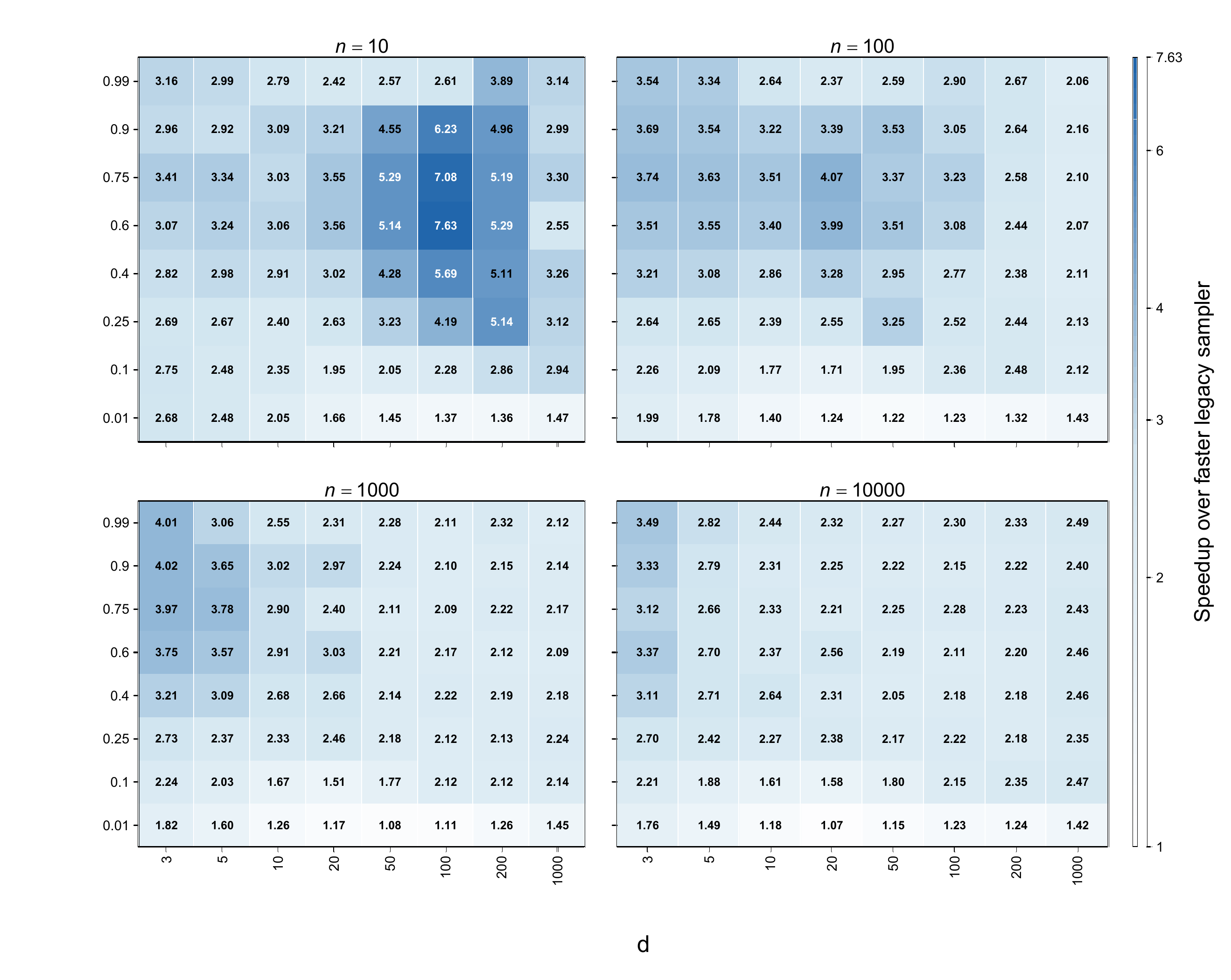}
\caption{Speed-up factor of the direct sampler relative to the better, in each
cell, of the angular central Gaussian and projected Saw samplers of
\citet{SablicaHornikLeydold2023}, for sample sizes $n=10,100,1000,10 000$. Each
panel shows dimension $d$ on the horizontal axis and concentration $\rho$ on the
vertical axis. Darker cells indicate larger speed-ups, and every cell exceeds
one.}
\label{fig:speedup}
\end{figure}

The direct sampler is faster in every one of the $256$ cells, by factors between
$1.1$ and $7.6$. The structure of the gains is informative and follows from the
two competitors having complementary weaknesses. The angular central Gaussian
sampler has a rejection constant that grows like $\sqrt{d}$ and, for fixed $\rho$,
behaves like $\sqrt{e}\,\rho(1-\rho^2)^{1/2}\sqrt{d}$ as $d\to\infty$
\citep{SablicaHornikLeydold2023}. The factor $\rho(1-\rho^2)^{1/2}$ vanishes as
$\rho\to0$ and as $\rho\to1$ and is largest near $\rho=1/\sqrt{2}$, so this
envelope is nearly exact at the extremes of the concentration range and least
efficient in its middle. The projected Saw sampler, by contrast, has a bounded
rejection constant but pays a univariate set-up cost that is essentially
independent of $n$, so its per-draw cost is dominated by set-up at small $n$ and
becomes negligible at large $n$. The better competitor is therefore usually the angular
central Gaussian sampler at small $n$ and the projected Saw sampler at large
$n$. Table~\ref{tab:pkbd-timings-n10} reports the corresponding average times
for $n=10$.

\begin{table}[!htbp]
\centering
\setlength{\tabcolsep}{3.5pt}
\renewcommand{\arraystretch}{0.92}

\caption{Average times in milliseconds for $n=10$, for different dimensions $d$
and parameters $\rho$. The rows compare the new direct sampler with the two legacy
samplers, ACG and Tinflex/Saw.}
\label{tab:pkbd-timings-n10}

\begin{tabular}{
  @{} c l *{8}{D{.}{.}{3}} @{}
}
\toprule
& & \multicolumn{8}{c}{$d$} \\
\cmidrule(l){3-10}
$\rho$ & Sampler
& \multicolumn{1}{c}{3}
& \multicolumn{1}{c}{5}
& \multicolumn{1}{c}{10}
& \multicolumn{1}{c}{20}
& \multicolumn{1}{c}{50}
& \multicolumn{1}{c}{100}
& \multicolumn{1}{c}{200}
& \multicolumn{1}{c}{1000} \\
\midrule

\multirow{3}{*}{$0.01$}
& Direct (ours)         & 0.004 & 0.004 & 0.005 & 0.008 & 0.015 & 0.026 & 0.060 & 0.24 \\
& ACG         & 0.010 & 0.011 & 0.011 & 0.013 & 0.021 & 0.035 & 0.081 & 0.35 \\
& Tinflex/Saw & 0.13  & 0.16  & 0.16  & 0.17  & 0.18  & 0.21  & 0.40  & 0.66 \\
\addlinespace[0.25em]

\multirow{3}{*}{$0.1$}
& Direct (ours)         & 0.004 & 0.004 & 0.006 & 0.008 & 0.015 & 0.026 & 0.050 & 0.24 \\
& ACG         & 0.010 & 0.011 & 0.013 & 0.016 & 0.030 & 0.059 & 0.14  & 1.30 \\
& Tinflex/Saw & 0.13  & 0.15  & 0.16  & 0.17  & 0.18  & 0.21  & 0.27  & 0.71 \\
\addlinespace[0.25em]

\multirow{3}{*}{$0.25$}
& Direct (ours)         & 0.004 & 0.004 & 0.006 & 0.008 & 0.015 & 0.026 & 0.050 & 0.25 \\
& ACG         & 0.010 & 0.012 & 0.014 & 0.021 & 0.048 & 0.11  & 0.29  & 2.94 \\
& Tinflex/Saw & 0.13  & 0.17  & 0.16  & 0.18  & 0.19  & 0.22  & 0.26  & 0.77 \\
\addlinespace[0.25em]

\multirow{3}{*}{$0.4$}
& Direct (ours)         & 0.004 & 0.004 & 0.005 & 0.008 & 0.015 & 0.026 & 0.050 & 0.24 \\
& ACG         & 0.011 & 0.013 & 0.015 & 0.024 & 0.063 & 0.15  & 0.43  & 4.52 \\
& Tinflex/Saw & 0.13  & 0.16  & 0.15  & 0.16  & 0.18  & 0.21  & 0.25  & 0.79 \\
\addlinespace[0.25em]

\multirow{3}{*}{$0.6$}
& Direct (ours)         & 0.004 & 0.004 & 0.005 & 0.008 & 0.015 & 0.026 & 0.051 & 0.29 \\
& ACG         & 0.012 & 0.014 & 0.016 & 0.029 & 0.076 & 0.20  & 0.55  & 5.82 \\
& Tinflex/Saw & 0.14  & 0.15  & 0.16  & 0.17  & 0.19  & 0.20  & 0.27  & 0.74 \\
\addlinespace[0.25em]

\multirow{3}{*}{$0.75$}
& Direct (ours)         & 0.004 & 0.004 & 0.006 & 0.008 & 0.015 & 0.030 & 0.051 & 0.24 \\
& ACG         & 0.013 & 0.014 & 0.017 & 0.029 & 0.078 & 0.21  & 0.56  & 5.96 \\
& Tinflex/Saw & 0.16  & 0.16  & 0.16  & 0.16  & 0.19  & 0.34  & 0.26  & 0.80 \\
\addlinespace[0.25em]

\multirow{3}{*}{$0.9$}
& Direct (ours)         & 0.004 & 0.004 & 0.005 & 0.008 & 0.015 & 0.030 & 0.052 & 0.26 \\
& ACG         & 0.012 & 0.013 & 0.017 & 0.026 & 0.067 & 0.18  & 0.45  & 4.83 \\
& Tinflex/Saw & 0.16  & 0.17  & 0.18  & 0.18  & 0.19  & 0.33  & 0.26  & 0.77 \\
\addlinespace[0.25em]

\multirow{3}{*}{$0.99$}
& Direct (ours)         & 0.004 & 0.004 & 0.005 & 0.008 & 0.015 & 0.032 & 0.050 & 0.24 \\
& ACG         & 0.012 & 0.013 & 0.015 & 0.019 & 0.038 & 0.085 & 0.19  & 1.80 \\
& Tinflex/Saw & 0.20  & 0.20  & 0.20  & 0.19  & 0.21  & 0.35  & 0.28  & 0.76 \\

\bottomrule
\end{tabular}
\end{table}

These two facts explain Figure~\ref{fig:speedup} and Table~\ref{tab:pkbd-timings-n10}. At small $n$ the comparison is mostly against the angular central Gaussian sampler, and the speed-up tracks its
rejection constant. It is modest near $\rho=0$ and $\rho=1$, where that sampler
is almost exact, and large in the middle of the concentration range and in
higher dimensions, where its $\sqrt{d}$ rejection cost bites. The largest
factor, $7.6$, occurs at $n=10$, $d=100$ and $\rho=0.6$. At large \(n\), the comparison is mostly with the projected Saw sampler, whose set-up is amortized. The speed-up remains about two across the grid, reflecting a smaller \(O(d)\) constant and no rejection loop. Because the direct sampler has neither a rejection loop nor a set-up phase, it beats whichever competitor is locally stronger, most decisively at moderate concentration, where neither competitor is near its best.

\section{Discussion}

The Poisson kernel-based distribution joins the short list of spherical
distributions that can be sampled exactly and in closed form. The generator
needs no envelope, no rejection loop, no tuning and no set-up, and its cost is
linear in the dimension, so it should be the default choice in the simulation,
bootstrap and Bayesian computations for which the distribution was introduced.
The same construction supplies the Newton kernel-based distribution as a
constant-free companion.

The result is therefore best viewed as a statistical use of a classical
chord geometry. It replaces proposal-based simulation by a direct
construction from one random line and one Bernoulli draw, and it explains why
the PKB and NKB distributions are exceptional members of the power-Saw
family. Other power-Saw exponents require additional weighting and are not
pursued here. The \texttt{rPKB} package and a script reproducing the experiments
and Figure~\ref{fig:speedup} are available
in a public \href{https://github.com/lsablica/rPKB}{GitHub repository}.

\section*{Funding}
The work of Lukas Sablica was supported by the Anniversary Fund of the
Oesterreichische Nationalbank (OeNB), grant number 19079.

\bibliographystyle{agsm}
\bibliography{references}
\end{document}